\documentclass[letterpaper, 10 pt, conference,usenames,dvipsnames]{ieeeconf}  

\IEEEoverridecommandlockouts                              

\usepackage{graphics} 
\usepackage{epsfig}
\usepackage{pgfplots} 
\pgfplotsset{plot coordinates/math parser=false} 
\newlength\figureheight 
\newlength\figurewidth 

\usepackage{bm}
\usepackage{amsmath,amssymb,amsthm}
\usepackage{mathtools}

\usepackage{algorithm}
\usepackage{algpseudocode}

\usepackage{tikz}\usetikzlibrary{arrows}
\usepackage{subfig}

\theoremstyle{plain}
\newtheorem{theorem}{Theorem}[section]
\theoremstyle{plain}
\newtheorem{proposition}[theorem]{Proposition}
\theoremstyle{plain}
\theoremstyle{definition}
\newtheorem{definition}[theorem]{Definition}
\theoremstyle{plain}

\theoremstyle{plain}
\newtheorem{corollary}[theorem]{Corollary}
\theoremstyle{plain}
\newtheorem{remark}[theorem]{Remark}
\newtheorem{assumption}[theorem]{Assumption}

\newcommand{\RR}{\mathbb{R}}

\newcommand{\mSetSymMat}[1]{S^{#1}}
\newcommand{\mSetPosSemSymMat}[1]{S^{#1}_+}
\newcommand{\mSetPosSymMat}[1]{S^{#1}_{++}}
\newcommand{\mIntInt}[2]{\mathcal I_{[#1, #2]}}
\newcommand{\mIntGeq}[1]{\mathcal I_{\geq #1}}

\newcommand{\mStateInputBound}{\mathcal W}

\newcommand{\mNormGen}[1]{\left\lVert {#1} \right\rVert}
\newcommand{\mDefFunction}[3]{#1: #2 \rightarrow #3}

\newcommand{\XX}{\mathbb{X}}
\newcommand{\UU}{\mathbb{U}}
\newcommand{\mSafeSet}{\mathcal{S}}
\newcommand{\mBackup}{\mathcal{B}}

\newcommand{\mConvexHull}[1]{\mathrm{co}(#1)}

\newcommand{\mDef}{\coloneqq}

\newcommand{\mTubeControl}{K_\Omega}
\newcommand{\mZpred}{{z}}
\newcommand{\mZpredOpt}{{z}^*}
\newcommand{\mVpred}{{v}}
\newcommand{\mVpredOpt}{{v}^*}

\newcommand{\mFeasibleSet}{\mathcal X_N}

\newcommand{\MM}{\mathcal M}

\newcommand{\LL}{\mathcal{L}}

\newcommand{\mCol}[1]{\mathrm {col}_#1}
\newcommand{\mRow}[1]{\mathrm {row}_#1}

\makeatletter
\newif\ifmygrid@coordinates
\tikzset{/mygrid/step line/.style={line width=0.80pt,draw=gray!80},
	/mygrid/steplet line/.style={line width=0.25pt,draw=gray!80}}
\pgfkeys{/mygrid/.cd,
	step/.store in=\mygrid@step,
	steplet/.store in=\mygrid@steplet,
	coordinates/.is if=mygrid@coordinates}
\def\mygrid@def@coordinates(#1,#2)(#3,#4){%
	\def\mygrid@xlo{#1}%
	\def\mygrid@xhi{#3}%
	\def\mygrid@ylo{#2}%
	\def\mygrid@yhi{#4}%
}
\newcommand\DrawGrid[3][]{%
	\pgfkeys{/mygrid/.cd,coordinates=true,step=1,steplet=0.2,#1}%
	\draw[/mygrid/steplet line] #2 grid[step=\mygrid@steplet] #3;
	\draw[/mygrid/step line] #2 grid[step=\mygrid@step] #3;
	\mygrid@def@coordinates#2#3%
	\ifmygrid@coordinates%
		\draw[/mygrid/step line]
		\foreach \xpos in {\mygrid@xlo,...,\mygrid@xhi} {%
				(\xpos,\mygrid@ylo) -- ++(0,-3pt)
				node[anchor=north] {$\xpos$}
			}
		\foreach \ypos in {\mygrid@ylo,...,\mygrid@yhi} {%
				(\mygrid@xlo,\ypos) -- ++(-3pt,0)
				node[anchor=east] {$\ypos$}
			};
	\fi%
}
\makeatother

\title{\LARGE \bf
	Linear model predictive safety certification for learning-based control
}
\author{Kim P. Wabersich and Melanie N. Zeilinger
\thanks{Kim Wabersich ({\tt\small wabersich@kimpeter.de}) and Melanie N. Zeilinger
({\tt\small mzeilinger@ethz.ch}) are with tfhe Institute for Dynamic Systems and Control,
ETH Zurich, Switzerland. This work was supported by the Swiss National Science Foundation under grant no. PP00P2 157601/1.}
}

\makeatletter
\newcommand\fs@betterruled{%
  \def\@fs@cfont{\bfseries}\let\@fs@capt\floatc@ruled
  \def\@fs@pre{\vspace*{5pt}\hrule height.8pt depth0pt \kern2pt}%
  \def\@fs@post{\kern2pt\hrule\relax}%
  \def\@fs@mid{\kern2pt\hrule\kern2pt}%
  \let\@fs@iftopcapt\iftrue}
\floatstyle{betterruled}
\restylefloat{algorithm}
\makeatother

\newcommand\copyrighttext{%
	\footnotesize \textbf{Published in: 2018 IEEE Conference on Decision and Control (CDC), DOI: 10.1109/CDC.2018.8619829.}\\
	\textcopyright 2018 IEEE. Personal use of this material is permitted. Permission from IEEE must be obtained for all other uses, in any current or
	future media, including reprinting/republishing this material for advertising or promotional purposes, creating new collective works,
	for resale or redistribution to servers or lists, or reuse of any copyrighted component of this work in other works.}
\newcommand\copyrightnotice{%
	\begin{tikzpicture}[remember picture,overlay]
		\node[anchor=south,yshift=3pt] at (current page.south) {\fbox{\parbox{\dimexpr\textwidth-\fboxsep-\fboxrule\relax}{\copyrighttext}}};
	\end{tikzpicture}%
}

\begin{document}

\maketitle
\copyrightnotice
\thispagestyle{empty}
\pagestyle{empty}

\begin{abstract}
\noindent
While it has been repeatedly shown that learning-based controllers can provide
superior performance, they often lack of safety guarantees.
This paper aims at addressing this problem by introducing a model predictive
safety certification (MPSC) scheme 
for linear systems with additive disturbances.
The scheme verifies safety of a proposed learning-based input and modifies
it as little as necessary in order to keep the system within a given set of
constraints. Safety is thereby related to the existence of a model predictive controller (MPC)
providing a feasible trajectory towards a safe target set. A robust MPC formulation
accounts for the fact that the model is generally uncertain in the context of learning, which
allows for proving constraint satisfaction at all times under the proposed MPSC strategy.
The MPSC scheme can be used in order to expand any
potentially conservative set of safe states and
we provide an iterative technique for enlarging the safe set.
Finally, a practical data-based design procedure for MPSC is
proposed using scenario optimization. 

\end{abstract}


\section{Introduction}

Learning-based control introduces new ways for controller synthesis
based on large-scale databases providing cumulated system knowledge.
This allows for previously intense tasks, such as system modeling and controller tuning,
to eventually be fully automated.
For example, deep reinforcement learning provides prominent results, with applications including
control of humanoid robots in complex environments \cite{merel2017humanBehavior}
and playing Atari Arcade video-games \cite{mnih2015atari}. 

Despite the advances in research-driven applications, the results can often not be transferred to
industrial systems that are \emph{safety-critical}, i.e. that must
be guaranteed to operate in a given range of physical and safety constraints. This is due
to the often complex functioning of learning-based methods rendering their systematic analysis difficult.

By introducing a model predictive safety certification (MPSC) mechanism
for any learning-based controller, we aim at bridging this gap for linear systems
with additive uncertainties that can, e.g., result from a belief representation of an unknown 
non-linear system. The proposed MPSC scheme estimates safety of a proposed
learning-based input in real-time by searching for a safe back-up trajectory
for the next time step in the form of generating a feasible trajectory towards a known safe set. Allowing the MPSC scheme
to modify the potentially unsafe learning-based input, if necessary, provides safety for all future times. 
The result can be seen as a `safety filter', since it only filters proposed inputs that drive
the system out of what we call the safe set. The resulting online optimization problem can
be efficiently solved in real-time using established model predictive (MPC) solvers. Partially unknown larger-scale
systems can therefore be efficiently enhanced with safety certificates during learning.

\emph{Contributions:}
We consider linear systems with additive disturbances, described in Section~\ref{sec:problem_description},
that encode the current, possibly data-driven, belief about a safety-critical system to which a potentially unsafe learning-based controller
should be applied. A model predictive safety certification scheme is proposed in Section III, which allows
for enhancing any learning-based controller with safety guarantees\footnote{Even human inputs can be enhanced
by the safety certification scheme, which relates e.g. to the concept of electronic stabilization control from
automotive engineering.}. The concept of the proposed scheme is comparable to the safety
frameworks presented in \cite{akametalu2014reachabilitySafety,wabersich2017scalableSafety} by providing an
implicit safe set together with a safe backup controller that can be applied if the system would
leave the safe set using the proposed learning input.
A distinctive advantage compared to existing methods is that the MPSC scheme can build on any system
behavior that is known to be safe, i.e. a known set of safe system states can be easily incorporated in our scheme
such that it will only analyze safety outside of the provided safe set.

The approach relies on scalable offline computations and online optimization of a robust MPC problem at every sampling time,
which can be performed using available real-time capable solvers that can deal with large-scale systems
(see e.g. \cite{domahidi2012efficientInteriorPoint}).
While we relate the required assumptions and design steps to tube-based MPC in Section
\ref{sec:calculation_omega_safe_set}, we present an automated, parametrization free, and data-driven design procedure,
that is tailored to the context of learning the system dynamics. The design
procedure and MPSC scheme are illustrated in Section~\ref{sec:numerical_examples} using numerical examples.

\emph{Related work:}
Making the relevant class of safety critical systems accessible to learning-based
control methods has gained significant attention in recent years. In addition to
the individual construction of safety certificates for specific learning-based
control methods subject to different notions of safety, see e.g. the survey
\cite{garcia2015aComprehensiveSurveySafeReinforcementLearning},
a discussion of advances in safe learning-based control
subject to state and input constraints can be found in \cite{fisac2017generalSafetyFramework}.
A promising direction that emerged from recent research focuses on what is called a `safety framework'
\cite{akametalu2014reachabilitySafety,fisac2017generalSafetyFramework,wabersich2017scalableSafety,larsen2017safeLearningDistributed},
which consists of a safe set in the state space and a safety controller.
While the system state is contained in the safe set, any feasible input
(including learning-based controllers) can be applied
to the system. However, if such an input would cause the system to leave the
safe set, the safety controller interferes. Since this strategy is compatible
with any learning-based control algorithm, it serves as a universal
safety certification concept.
The techniques proposed in \cite{akametalu2014reachabilitySafety,fisac2017generalSafetyFramework}
are based on a differential game formulation that results in solving
a min-max optimal control problem, which can provide the largest possible safe set,
but offers very limited scalability.
The approach described in \cite{wabersich2017scalableSafety} uses convex approximation
techniques that scale well to larger-scale systems at the cost
of a potentially conservative safe set. While these results explicitly consider
non-linear systems, we focus on linear model approximations allowing for various improvements.
We introduce a new mechanism for generating the safe set and controller using
ideas related to tube-based MPC, which enables scalability
with respect to the state dimension, while being less conservative than
e.g. \cite{wabersich2017scalableSafety}.

There is a methodological
similarity to learning-based MPC approaches, as e.g. proposed in
\cite{Aswani2013}, or more recently in \cite{Koller2018} considering
nonlinear Gaussian process models. While such methods are limited to an
MPC strategy based on the learned system model, this paper provides a concept
that can enhance any learning-based controller with safety guarantees.
This allows, e.g., for maximizing black-box reward functions
(reward of a sequence of actions is only available through measurements)
for complex tasks, see e.g. \cite{Mania2018}, which would not be
possible within an MPC framework, or to focus on exploration
in order to collect informative data about the system,
as described in Section~\ref{sec:numerical_examples}.

\emph{Notation:}
The set of symmetric matrices of dimension $n$
is $\mSetSymMat{n}$, the set of positive (semi-)
definite matrices is ($\mSetPosSemSymMat{n}$) $\mSetPosSymMat{n}$, 
the set of integers in the interval $[a,b]\subset\RR$ is
$\mIntInt{a}{b}$, and the set of integers in the interval
$[a,\infty)\subset\RR$ is $\mIntGeq{a}$.
The Minkowski sum of two sets $\mathcal A_1, \mathcal A_2 \subset \RR$
is denoted by $\mathcal A_1 \oplus \mathcal A_2$ and the
Pontryagin set difference by $\mathcal A_1 \ominus \mathcal A_2$.
The $i$-th row and $i$-th column of
a matrix $A\in\RR^{n\times m}$ is denoted by $\mRow{i}(A)$
and $\mCol{i}(A)$.


\section{Problem description}\label{sec:problem_description}
We consider dynamical systems, which can be described by
linear systems with additive disturbances of the form
\begin{align}\label{eq:linear_system_additive_disturbance}
	x(k+1) = Ax(k) + Bu(k) + w(k)
\end{align}
with initial condition $x(0)=x_0$ and $w(k)\in\mStateInputBound$ where $\mStateInputBound$
is a compact set.
The system is subject to polytopic state constraints $x(k) \in \XX
\mDef \lbrace x\in \RR^n|A_x x \leq b_x \rbrace$, 
$A_x\in\RR^{n_x\times n}$, $b_x\in\RR^{n_x}$ and polytopic input
constraints $u(k)\in\UU\mDef\lbrace u\in\RR^m|A_u u \leq b_u \rbrace$,
$A_u\in\RR^{n_u\times m}$, $b_u\in\RR^{n_u}$.
We assume that the origin is contained in $\XX$, $(A,B)$ is stabilizable,
and the system state is fully observable.
Note that system class \eqref{eq:linear_system_additive_disturbance} allows for modeling
nonlinear time-varying systems
	$x(k+1) = f(k,x(k),u(k))$ 
if $x(k+1) \in Ax(k) + Bu(k) \oplus \mStateInputBound$ for all $(x,u)\in (\XX,\UU)$.

We aim at providing a safety certificate for arbitrary control signals in terms of a
safe set and a safe control law. Given the system description \eqref{eq:linear_system_additive_disturbance}
and a potentially unsafe learning-based controller $u_\LL$,
we search for a set of states $\mSafeSet$ for which we know a feasible backup control strategy $u_\mBackup$
such that input and state constraints will be fulfilled for all future times. Therefore, 
$u_\LL$ can be applied as long as it does not cause the system to leave $\mSafeSet$ or 
violate input constraints.
Otherwise, a safety controller $u_\mSafeSet$ is allowed
to modify the learning input based on the backup controller in order to keep the system safe. 
Formally this is captured by the following definition of a safe set and controller.

\begin{definition}\label{def:safe_set}
	A set $\mSafeSet\subseteq \XX$ is called a \emph{safe set} for
	system \eqref{eq:linear_system_additive_disturbance} if a 
	\emph{safe backup control law} $\mDefFunction{u_\mBackup}{\RR^n\times\RR^m\times\mIntGeq{0}}{\UU}$
	is available such that for an arbitrary (learning-based) policy
	$\mDefFunction{u_\LL}{\mIntGeq{0}}{\RR^m}$, the application of the \emph{safety
	control law}
	\begin{align*}
			u_\mSafeSet(k)\mDef \begin{cases}
				u_\LL(k),  ~ \text{if}~u_\LL\in\UU \land \{ Ax + Bu_\LL \} \oplus \mStateInputBound \subseteq \mSafeSet\\
				u_\mBackup(x(k) ,u_\LL(k),k), ~ \text{otherwise} 
			\end{cases}
	\end{align*}
	guarantees that the system state $x(k)$ is contained in $\XX$ for
	all $k \geq \bar k$ if $x(\bar k)\in\mSafeSet$.
\end{definition}
While this framework is conceptually similar to those in 
\cite{akametalu2014reachabilitySafety,fisac2017generalSafetyFramework,wabersich2017scalableSafety,larsen2017safeLearningDistributed},
we do not require the safe set $\mSafeSet$ to be robust controlled invariant as in
\cite[Definition II.4]{wabersich2017scalableSafety},
\cite[Definition 2]{fisac2017generalSafetyFramework}, \cite[Section 2.2]{larsen2017safeLearningDistributed} or
\cite[Section II.A]{akametalu2014reachabilitySafety}.
The presented approach is thereby capable of enlarging any given safe set, and
can be combined with any of the previously proposed methods.



\section{Model predictive safety certification}\label{sec:mpsc}

\begin{figure}
	\centering
	\vspace{0.45cm}
	\begin{tikzpicture}[scale=0.6]
		\input{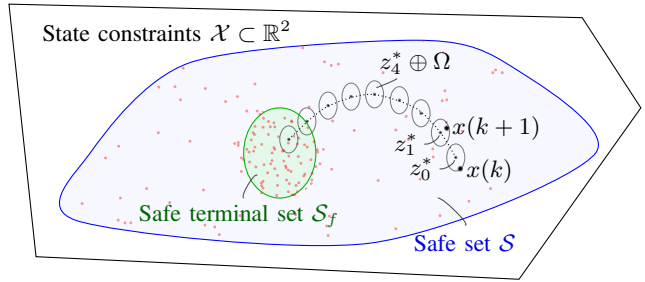}
	\end{tikzpicture}
	\caption{\small Model predictive safe set (blue) based on
	available data points (red): Belief-based safe optimal
	trajectory that supports $u_\LL(k)$, starting nearby $x(t)$ and resulting in the safe terminal set,
	with uncertainty tubes, that will contain the real system state.}
	\label{fig:mode_predictive_safety}
\end{figure}

The starting point for the derivation of the proposed safety concept is Definition \ref{def:safe_set}.
The essential requirement is that in the safe set $\mSafeSet$, we always need to know a feasible backup
controller $u_\mBackup$, that ensures constraint satisfaction in the face of uncertainty
for all future times. The idea for constructing such a controller is based on MPC
\cite{rawlings2009model}. Given the current system
state, we calculate a safe, finite horizon backup controller towards some conservative
target set $\mSafeSet_f$, which is known to be a safe set
and therefore provides `infinite safety' after applying the finite-time controller.

The concept is illustrated in Figure \ref{fig:mode_predictive_safety}.
Consider a current system state $x(k)$, together with a proposed learning
input $u_\LL(k)$. 
In order to analyze safety of $u_\LL(k)$, we test if $u_\LL(k)$
will lead us to a state $x(k+1)$, for which we can construct
a safe backup controller $u_\mBackup$ in the form of a feasible input sequence
that drives the system to the safe terminal set $\mSafeSet_f$ in a given finite number of steps.
If the test is successful, $x(k+1)$ is a safe state and $u_\LL(k)$ can be applied. 
At the next time step $k+1$, we repeat the calculations
for $x(k+1)$ and $u_\LL(k+1)$. If it is successful, we can again apply
the learning input $u_\LL(k+1)$, otherwise we can simply use the previously calculated
backup controller from time $k$. This strategy yields a safe set that is defined
by the feasible set of the corresponding optimization problem for planning a trajectory
towards the target set.

As the true system dynamics model is often unknown in the context of learning-based control,
we employ mechanisms from tube-based MPC to design a safe backup controller
for uncertain system dynamics of the form \eqref{eq:linear_system_additive_disturbance}.

\subsection{Model predictive safety certification scheme}

Similar as in tube-based MPC, see e.g. \cite{rawlings2009model}, a nominal backup trajectory is
computed, such that a stabilizing auxiliary controller is
able to track it for the real system within a `tube' towards the safe terminal set.
We first define the main components and assumptions of the tube-based MPC controller,
in order to then introduce the \emph{model predictive safety certification} (MPSC) scheme,
consisting of the MPSC problem and the proposed safety controller.
Define with $z(k)\in \RR^n$ and $v(k)\in\RR^m$ the nominal system
states and inputs, as well as the nominal dynamics
\begin{align}\label{eq:nominal_system}
	z(k+1) = Az(k) + Bv(k), ~k\in\mIntGeq{0}
\end{align}
with initial condition $z(0)=z_0$.
Denote $e(k)\mDef x(k)-z(k)$ as the error (deviation) between the system state
\eqref{eq:linear_system_additive_disturbance} and the nominal system state
\eqref{eq:nominal_system}. The controller is then defined by augmenting
the nominal input with an auxiliary feedback on the error, i.e.
\begin{align}\label{eq:auxiliary_control}
	u(k) = v(k) + \mTubeControl(x(k)-z(k)),
\end{align}
which keeps the real system state $x(k)$ close to the nominal system state $z(k)$
if $\mTubeControl\in\RR^{m\times n}$ is chosen such that it robustly stabilizes the
error $e(k) \mDef x(k) - z(k)$ with dynamics
\begin{align}\label{eq:error_dynamics}
	e(k+1) = (A + BK_\Omega) e(k) + w(k)
\end{align}
resulting from application of \eqref{eq:auxiliary_control} to the real system.

\begin{assumption}\label{ass:existence_RPI}
	There exists a linear state feedback matrix $\mTubeControl\in\RR^{m\times n}$
	that yields a stable error system \eqref{eq:error_dynamics}.
\end{assumption}

Stability of the autonomous error dynamics \eqref{eq:error_dynamics} implies the 
existence of a corresponding robust positively invariant set according to the
following definition.
\begin{definition}\label{def:RPI}
	A set $\Omega\subseteq\RR^n$ is a robust positively invariant (RPI) set
	for the error dynamics \eqref{eq:error_dynamics} if
	\begin{align*}
		e(k_0) \in \Omega \Rightarrow e(k) \in \Omega \text{ for all } k\in\mIntGeq{k_0}.
	\end{align*}
\end{definition}

In order to guarantee that $(x,u)\in \XX\times\UU$
under application of \eqref{eq:auxiliary_control}, the
state and input constraints $\XX$ and $\UU$ are tightened for the nominal system \eqref{eq:nominal_system},
as described e.g. in \cite{rawlings2009model}, to $\bar \XX = \XX \ominus \Omega$
and $\bar \UU = \UU \ominus K_\Omega\Omega$.
There exist various methods in the literature, which can be used in order to
calculate a controller and the corresponding RPI set according to
Definition \ref{def:RPI}, see e.g. \cite{rakovic2005invariant}.

Different from standard tube-based MPC, the model predictive safety certification (MPSC)
uses a terminal set that is only required to be itself a safe set
according to Definition~\ref{def:safe_set}, which is conceptually similar to the safe terminal
set used in \cite{Koller2018}. This allows not only for enlarging 
any potentially conservative initial safe set, but also for recursively improving 
the safe set, as will be shown in Section~\ref{sec:iterative_enlargement_of_Sf}.

\begin{assumption}\label{ass:terminal_set}
	There exists a safe set $\mathcal \mSafeSet_f\subseteq \XX$ and a safe control
	law $u_{\mSafeSet_f}$ according to Definition
	\ref{def:safe_set} such that $\Omega\subseteq\mSafeSet_f$.
\end{assumption}

\begin{remark}
	Note that a standard terminal set as, e.g. described in \cite{rawlings2009model}, also satisfies
	Assumption \ref{ass:terminal_set}. A trivial choice is therefore
	$\mathcal \mSafeSet_f = \Omega$ and $u_{\mSafeSet_f}(k) = \mTubeControl x(k)$.
\end{remark}

Based on these components, the proposed safe backup controller utilizes the following MPSC problem
for a given measured state $x$ and proposed learning input $u_\LL$:
\begin{subequations}\label{eq:MPSC_opt}
	\begin{align} \label{eq:MPSC_opt_cost}
			\min_{\substack{\mVpred_0,..,\mVpred_{N-1},\\\mZpred_0,..,\mZpred_{N}, \tilde u}} ~ & \mNormGen{u_\LL - \tilde u} \\ \label{eq:MPSC_opt_nominal_dynamics}
			\text{s.t.} ~~& \mZpred_{i+1} = A\mZpred_i + B\mVpred_i,~ \forall i\in\mIntInt{0}{N-1} \\ \label{eq:MPSC_opt_tightening}
						 & (\mZpred_i, \mVpred_i)\in \bar \XX \times \bar \UU,~ \forall i\in\mIntInt{0}{N-1} \\ \label{eq:MPSC_opt_terminal}
						 & \mZpred_N \in \mathcal \mSafeSet_f \ominus \Omega  \\ \label{eq:MPSC_safety_constraint}
						 & x - z_0  \in \Omega\\ \label{eq:MPSC_aux_vars_constraints}
						 & \tilde u  = \mVpred_0 + K_\Omega(x - \mZpred_0)
	\end{align}
\end{subequations}
where we denote the planning horizon by $N\in\mIntGeq{1}$ and the predicted nominal system states and
inputs by $\mZpred_i$ and $\mVpred_i$.
Let the feasible set of \eqref{eq:MPSC_opt} be denoted by
\begin{align}\label{eq:feasible_mpsF_set}
	\mFeasibleSet \mDef \{ x\in \RR^n |
		\eqref{eq:MPSC_opt_nominal_dynamics}-\eqref{eq:MPSC_aux_vars_constraints} \}\subseteq \XX.
\end{align}
Problem \eqref{eq:MPSC_opt} introduces the auxiliary variable $\tilde u$, which includes the auxiliary feedback
\eqref{eq:MPSC_aux_vars_constraints}, ensuring safety of the control input
$\tilde u$, as we will show in the proof of Theorem~\ref{thm:mpsf_safe}.
The cost \eqref{eq:MPSC_opt_cost} is chosen such that if possible, $\tilde u$ is equal to $u_\LL$,
in which case safety of $u_\LL$ is certified.
The controller resulting from \eqref{eq:MPSC_opt} in a receding horizon fashion is given by 
\begin{align}\label{eq:implicit_safe_control_law}
	\kappa(x) = \tilde u^*(x) 
\end{align}
where $\mVpred_0^*(x),..,\mVpred_{N-1}^*(x)$, $\mZpred_0^*(x),..,\mZpred_{N}^*(x)$, and $\tilde u^*(x)$
is the optimal solution of \eqref{eq:MPSC_opt} at state $x$.

It is important to note that \eqref{eq:MPSC_opt} may not be recursively feasible for general
safe sets $\mSafeSet_f$.
This is due to the fact that the terminal safe set $\mSafeSet_f$ is itself
not necessarily invariant or a subset of the feasible set.

To this end, we propose Algorithm \ref{alg:MPSC}, which implements a safety controller based on \eqref{eq:MPSC_opt}.
If $x(k)\in\mFeasibleSet$ we can always directly apply $\tilde u$ (Algorithm \ref{alg:MPSC}, line 4).
If at a subsequent time $x(k+1)\notin\mFeasibleSet$, then we know via \eqref{eq:MPSC_opt} a finite-time safe
backup controller towards $\mSafeSet_f$ using \eqref{eq:auxiliary_control}, from the trajectory computed to certify $x(k)$
(Algorithm \ref{alg:MPSC}, line 9), compare also with Figure \ref{fig:mode_predictive_safety}.
By Assumption \ref{ass:terminal_set} we can extend this finite-time backup controller after $N$-steps with $u_{\mSafeSet_f}$
(Algorithm \ref{alg:MPSC}, line 11) in order to obtain a safe backup controller for $x(k+i), i\geq 0$,
which will satisfy constraints at all times in the future.
In the case that $\mSafeSet_f\setminus\mFeasibleSet\neq\emptyset$, \eqref{eq:MPSC_opt} can be initially infeasible for
$x=x(0)\in\mSafeSet_f$. This case can be easily treated by directly applying $u_{\mSafeSet_f}$
(Algorithm \ref{alg:MPSC}, line 1 and line 11) which ensures
safety for all future times. Formalization of the above yields our main result.

\input{algorithms/MPSC.alg}
\begin{theorem}\label{thm:mpsf_safe}
	If Assumptions \ref{ass:existence_RPI} and \ref{ass:terminal_set} hold, then the control law
	resulting from Algorithm~\ref{alg:MPSC} is a safe backup controller and
	$\mFeasibleSet \cup \mSafeSet_f$ the corresponding safe set according to Definition \ref{def:safe_set}.
\end{theorem}
\begin{proof}
	If $x(0)\in\mSafeSet_f\setminus\mFeasibleSet$, the terminal safety controller $u_{\mSafeSet_f}$
	is applied since $k_\text{inf}$ is initialized to $N-1$ in (Algorithm~\ref{alg:MPSC}, line 1), which 
	keeps the system safe for all times.
	We show that $\mFeasibleSet$ is a safe set by first investigating the case that \eqref{eq:MPSC_opt}
	is feasible for all $k\in\mIntGeq{0}$ and extending the analysis to cases in which \eqref{eq:MPSC_opt}
	is infeasible for arbitrarily many time steps.

	Let $x(0)\in\mFeasibleSet$ and let \eqref{eq:MPSC_opt} be feasible
	for all $k\in\mIntGeq{0}$ (Algorithm~\ref{alg:MPSC}, line 4), i.e. $x(k) \in \mFeasibleSet$.
	Condition \eqref{eq:MPSC_safety_constraint} implies	by
	Assumption \ref{ass:existence_RPI} that $e(k+1)\in\Omega$ and therefore that $x(k+1)\in \mZpred_1 \oplus \Omega$,
	which implies by the tightened constraints on the nominal state \eqref{eq:MPSC_opt_tightening} that $x(k+1)\in\XX$.
	Therefore $\mFeasibleSet$ is a safe set under the safe backup controller $\kappa(x)$ in \eqref{eq:implicit_safe_control_law}.
	
	Now, consider an arbitrary time $\bar k$ for which $\eqref{eq:MPSC_opt}$ was feasible
	for the last time, i.e. $x(k)\notin \mFeasibleSet$ for all $k = \bar k + k_{\text{inf}}$ with
	$k_{\text{inf}}\in\mIntInt{1}{N-1}$.
	Because of \eqref{eq:MPSC_safety_constraint} and \eqref{eq:MPSC_aux_vars_constraints} we have
	that $e(\bar k+1) \in \Omega$ and therefore Assumption \ref{ass:terminal_set}
	together with \eqref{eq:MPSC_opt_terminal} allows for explicitly stating a safe backup control
	law based on $x=x(\bar k)\in\mFeasibleSet$ that keeps the system in the constraints for all future times:
	\begin{align*}
		&u_\mBackup(x(i),x,i) = \\                                                               
		&\quad \begin{cases}                                                                  
		\mVpred_{i-\bar k}^*(x) + \mTubeControl(x(i-\bar k) - \mZpred_{i-\bar k}^*(x))) ,  i\in\mIntInt{\bar k+1}{\bar k+N-1} \\
		u_{\mSafeSet_f}(i), i\in\mIntGeq{\bar k+N}.                                      
		\end{cases}                                                                    
	\end{align*}
	Since \eqref{eq:MPSC_opt} was feasible for
	$x(\bar k)$, the corresponding $\mVpredOpt_1,..,\mVpredOpt_{N-1}$, $\mZpredOpt_1,..,\mZpredOpt_{N}$ exist.
	Therefore in case of $x(k) \notin \mFeasibleSet$, Algorithm~\ref{alg:MPSC}, line 9
	and $k_{\text{inf}}\in\mIntInt{1}{N-1}$, it follows from 
	$e(\bar k + 1 )\in\Omega$ by Assumption \ref{ass:existence_RPI}
	that $x(\bar k + k_{\text{inf}})\in \mZpredOpt_{k_{\text{inf}}} \oplus \Omega$ for all $k_{\text{inf}}\in\mIntInt{1}{N-1}$.
	
	The last remaining case $k_{\text{inf}}\geq N$ follows from the observation, that
	$x(\bar k + N)\in \mSafeSet_f$ by \eqref{eq:MPSC_opt_terminal} for which we know the safe control law $u_{\mSafeSet_f}$
	for all future times by Assumption~\ref{ass:terminal_set}. Once a feasible solution is found again, the counter $k_\text{inf}$
	is set to zero.
	Consequently we investigated all possible cases in Algorithm~\ref{alg:MPSC} and proved that it will always
	provide a safe control input if $x(0)\in\mFeasibleSet \cup \mSafeSet_f$, showing the result.
\end{proof}




\subsection{A recursively feasible MPSC scheme}

By modifying Assumption \ref{ass:terminal_set} and requiring the terminal safe set to be invariant
for the nominal system, which is the standard assumption in tube-based MPC,
we obtain recursive feasibility of \eqref{eq:MPSC_opt} and can thus directly
apply the time-invariant control law \eqref{eq:implicit_safe_control_law} to
system \eqref{eq:linear_system_additive_disturbance} without the need of
Algorithm~\ref{alg:MPSC}. In other words, \eqref{eq:implicit_safe_control_law}
directly becomes the safety controller according to Definition \ref{def:safe_set}.

\begin{assumption}\label{ass:terminal_configuration_tube_mpc}
	There exists a set $\mathcal X_f\subseteq\bar \XX$ and 
	corresponding control law $\mDefFunction{\sigma_f}{\mathcal X_f}{\bar \UU}$
	such that for all $z\in \mathcal X_f$, $\Rightarrow A z + B\sigma_f(z) \in \mathcal X_f$.
\end{assumption}

\begin{theorem}\label{thm:mpsf_safe_recursive_feasibility}
	Let $\mSafeSet_f = \mathcal X_f \oplus \Omega$. If Assumptions \ref{ass:existence_RPI},
	and \ref{ass:terminal_configuration_tube_mpc} hold, 
	then \eqref{eq:implicit_safe_control_law} is a safe backup control law and
	$\mFeasibleSet$ the corresponding safe set according to Definition \ref{def:safe_set}.
	In addition, $\mFeasibleSet$ is a robust positively invariant set.
\end{theorem}

\begin{proof}
	We begin with showing recursive feasibility under \eqref{eq:implicit_safe_control_law}.
	Let \eqref{eq:MPSC_opt} be feasible at time $k$. It follows that $x(k+1)\in \mZpredOpt_1 \oplus \Omega$
	because of \eqref{eq:MPSC_safety_constraint} and \eqref{eq:MPSC_aux_vars_constraints}. From here, recursive feasibility
	follows as in standard tube-based MPC by induction, see e.g. \cite{rawlings2009model}.
	Along the lines of the proof of Theorem~\ref{thm:mpsf_safe}, recursive
	feasibility implies that $\mFeasibleSet$ is a safe set.
\end{proof}



\section{Design of $\mSafeSet_f$ and $\Omega$ from data}\label{sec:calculation_omega_safe_set}

The proposed MPSC scheme is based on two main design components, the robust positively
invariant set $\Omega$, which determines the tube, and the terminal safe set $\mSafeSet_f$. 

While $\mSafeSet_f$ can be chosen more generally according to Assumption \ref{ass:terminal_set}, we note from
Theorem~\ref{thm:mpsf_safe_recursive_feasibility} that we can in principle also use
the same design methods proposed for linear tube-based MPC.
The computation of the robust invariant set $\Omega$ and
the nominal terminal set $\mathcal X_f$ have been widely studied in the literature, see e.g.
\cite{rawlings2009model,blanchini1999setInvariance}
and references therein.

This section presents a different option for the approximation of a tube and safe terminal set that is 
tailored to the learning context and aims at a minimal amount of tuning `by hand'. We propose
to infer a robust control invariant set $\Omega$ either directly from data or
from a probabilistic model via scenario-based optimization. Secondly,
starting from any terminal safe set, e.g. the trivial choice $\{0\}\oplus\Omega$, we show how to enlarge
this terminal set iteratively by utilizing feasible solutions of \eqref{eq:MPSC_opt}
over time.

\subsection{Scenario based calculation of $\Omega$ from data}\label{sec:computation_omega}

Let $\{\tilde w_i\}_{i=1}^{N_s}$ be a set of so-called `scenarios',
either sampled from a probabilistic belief about the system
dynamics \eqref{eq:linear_system_additive_disturbance}
or collected from measurements. We restrict ourselves to ellipsoidal robust positively invariant
sets $\Omega = \{ x|x^\top P x \leq 1 \}$ with $P\in\mSetPosSymMat{n}$, in order to enable scalability of the
resulting design optimization problems to larger scale systems. The corresponding robust scenario-based
design problem for computation of the set $\Omega$ is given by
\begin{subequations}\label{eq:scenario_based_omega_opt}
\begin{align}\label{eq:scenario_based_omega_opt_objective}
	& \min_{P\in\mSetPosSymMat{n},\tau >0}  ~ -\log \det (P) \\\nonumber
	& \text{s.t.}  ~ \forall i\in \mIntInt{1}{N_s}: \\\label{eq:scenario_based_omega_opt_lmi}
		~ & ~\begin{pmatrix}
			A_{cl}^\top P A_{cl} - \tau P & A_{cl}^\top P \tilde w_i \\
			\tilde w_i^\top P A_{cl} & \tilde w_i^\top P \tilde w_i + \tau - 1
		\end{pmatrix} \preceq 0
\end{align}
\end{subequations}
where $A_{cl} \mDef A + B \mTubeControl$.	
Problem \eqref{eq:scenario_based_omega_opt} defines a robust
positively invariant set for the error system \eqref{eq:error_dynamics},
if the condition is enforced for all $\tilde w_i \in\mStateInputBound$, see
e.g. \cite{blanchini1999setInvariance}.
The objective \eqref{eq:scenario_based_omega_opt_objective} is chosen such
that a possibly small RPI set is obtained, which increases by definition
of $\bar\XX$ and $\bar \UU$ the size of the feasible region of \eqref{eq:MPSC_opt},
and therefore the size of the safe set. A stabilizing linear state feedback
matrix $K_\Omega$ according to Assumption \ref{ass:existence_RPI} needs to
be chosen beforehand, e.g. using LQR or $\mathcal H_\infty$ controller
design methods. 
\begin{proposition}\label{prop:scenario_based_omega}
	Consider system \eqref{eq:error_dynamics}
	and let $n_s \mDef (n^2+n)/2 + 1$.
	If \eqref{eq:scenario_based_omega_opt} attains a solution,
	then with probability at least
	$1 - \sum_{i=0}^{n_s-1}\begin{psmallmatrix} N_s \\i \end{psmallmatrix}
		\epsilon^i(1-\epsilon)^{N_s-i}$,
	the solution is $\epsilon$-level robustly feasible for the
	corresponding robust problem imposing \eqref{eq:scenario_based_omega_opt_lmi}
	for all $\tilde w_i \in \mStateInputBound$, i.e., the probability that there exists a
	$\tilde w_i \in \mStateInputBound$ for which \eqref{eq:scenario_based_omega_opt_lmi}
	is violated is less or equal to $\epsilon$.
\end{proposition}
\begin{proof}
	The result follows directly from \cite[Theorem 1]{campi2008scenario} similar
	to the application demonstrated in \cite{calafiore2006scenario}.
\end{proof}


\subsection{Iterative enlargement of the terminal safe set $\mSafeSet_f$}\label{sec:iterative_enlargement_of_Sf}

In this section we show how to enlarge the terminal safe set $\mSafeSet_f$ based on previously
calculated solutions of \eqref{eq:MPSC_opt}, which is conceptually similar to the data-based
terminal set proposed in \cite{rosolia2017learning}. Note, that a larger terminal set $\mSafeSet_f$
according to Assumption \ref{ass:terminal_set} or Assumption \ref{ass:terminal_configuration_tube_mpc}
typically also leads to a larger feasible set $\mFeasibleSet$, and therefore to a larger
overall safe set $\mSafeSet$ according to Theorems \ref{thm:mpsf_safe} and
\ref{thm:mpsf_safe_recursive_feasibility}.

The main idea is to define a safe set based on successfully
solved instances of \eqref{eq:MPSC_opt} for measured system states
\begin{align}\label{eq:tuple_of_feasible_trajectories}
	\bm x^*_{\MM(k)} \mDef \left\{ x(i), i\in \MM(k) \right\}
\end{align}
where $\MM(k) \mDef \{ i \in\mIntInt{0}{k} | x(i) \in \mFeasibleSet \}$
is an index set representing time instances for which the system state $x(i)$
was feasible in terms of \eqref{eq:MPSC_opt} during
application of Algorithm~\ref{alg:MPSC} up to time $k$.

\begin{theorem}\label{thm:enlargement_of_terminal_set}
	If Assumptions \ref{ass:existence_RPI} and \ref{ass:terminal_set}
	are satisfied, and
	\eqref{eq:MPSC_opt} is convex, then the set
	\begin{align}\label{eq:enlargement_of_terminal_set}
		\mSafeSet_f^{\MM(k)} \mDef \mConvexHull{\bm x^*_{\MM(k)}} \cup \mSafeSet_f
	\end{align}
	is again a safe set according to Definition \ref{def:safe_set}
	with a safe backup controller given by Algorithm~\ref{alg:MPSC}.
\end{theorem}
\begin{proof}
	If \eqref{eq:MPSC_opt} is convex, then the feasible set is a convex set, see e.g. \cite{boyd2004convex}
	and therefore $\mConvexHull{\bm x^*_{\MM(k)}}\subseteq\mFeasibleSet$.
	As a consequence, we can solve \eqref{eq:MPSC_opt} for all $x\in\mConvexHull{\bm x^*_{\MM(k)}}$ which
	in turn provides the result	by the proof of Theorem~\ref{thm:mpsf_safe} (by replacing $\mFeasibleSet$ with
	$\mConvexHull{\bm x^*_{\MM(k)}}$) and the fact that the union of two safe sets is again a
	safe set.
\end{proof}

If $\mSafeSet_f \subseteq \mConvexHull{\bm x^*_{\MM(k)}}$, convexity of the new terminal set
\eqref{eq:enlargement_of_terminal_set} is ensured and
we can iteratively enlarge the initial terminal set $\mSafeSet_f$.

\begin{remark}\label{rem:practical_design_procedure}
	A practical design procedure in order to determine $\Omega$ and $\mSafeSet_f$ is as follows.
	Compute $\Omega$ based on measurements as described in
	Section~\ref{sec:computation_omega} and initialize $\mSafeSet_f = \{\Omega\}$. Then, during
	closed-loop operation of Algorithm~\ref{alg:MPSC} enlarge $\mSafeSet_f$ according to
	\eqref{eq:enlargement_of_terminal_set}.
\end{remark}

In order to provide a similar result with respect to Theorem~\ref{thm:mpsf_safe_recursive_feasibility}
consider the set $\bm z^*_{\MM(k)} = \{\mZpredOpt_1(x(i)),..,\mZpredOpt_N(x(i)), i\in\MM(k)\}$ with $\MM(k)$
as defined above.

\begin{corollary}\label{cor:enlargement_of_terminal_set_tube}
	If Assumptions \ref{ass:existence_RPI} and
	\ref{ass:terminal_configuration_tube_mpc} are satisfied, $\mSafeSet_f = \mathcal X_f\oplus\Omega$, and
	\eqref{eq:MPSC_opt} is convex, then the set
	\begin{align}\label{eq:enlargement_of_terminal_set_tube}
		\mathcal X_f^{\MM(k)} \mDef \mConvexHull{\bm z^*_{\MM(k)}} \cup \mathcal X_f 
	\end{align}
	satisfies Assumption \ref{ass:terminal_configuration_tube_mpc} and 
	$\mathcal X_f^{\MM(k)}\oplus\Omega$ is a safe set according to Definition
	\ref{def:safe_set} with safe backup controller \eqref{eq:implicit_safe_control_law}.
\end{corollary}
\begin{proof}
	Follows similarly to the proof of Theorem~\ref{thm:enlargement_of_terminal_set}.
\end{proof}

Using Theorem \ref{thm:mpsf_safe_recursive_feasibility}, we obtain a practical
procedure similar to Remark~\ref{rem:practical_design_procedure} by
initializing $\mathcal X_f = \{0\}$ and choosing $\mathcal X_f^{\MM(k)}\oplus\Omega$
as iterative terminal safe set. 


\begin{remark}
	Theorem~\ref{thm:enlargement_of_terminal_set} also provides an explicit approximation 
	of the safe set given by
	\eqref{eq:enlargement_of_terminal_set}, which is generally only implicitly defined.
	Such a representation can be used to `inform' the learning-based controller about the safety boundary,
	e.g. in the form of a feature using a barrier function, in order to avoid chattering behavior, as
	proposed in \cite{akametalu2014reachabilitySafety}.
\end{remark}


\section{Application to numerical examples}\label{sec:numerical_examples}
We consider the problem of safely acquiring information about the partially unknown
system dynamics of a discretized mass-spring-damper system, which is given by
$
	x(k+1) =
		\begin{psmallmatrix}
			1    & 0.1 \\
			-0.3 & 0.8
		\end{psmallmatrix}
	x(k) +
		\begin{psmallmatrix}
			0 \\
			0.1
		\end{psmallmatrix}
	u(k)
$
with $|u(k)|\leq 2.5$, $|x_1(k)| \leq 1$, and $x_2(k) \in [-0.4, 1]$.
Assume that an approximate model is given by
$
	x(k+1) =
	\begin{psmallmatrix}
		1 & 0.1 \\
		-0.23 & 0.78
	\end{psmallmatrix}
	x(k) +
	\begin{psmallmatrix}
		0 \\
		0.1
	\end{psmallmatrix}
	u(k)
	+ w(k)
$
with mass, spring, and damper parameters, which have a 20\% error with respect to the
true parameters. We use the results from Section~\ref{sec:calculation_omega_safe_set} in
order to calculate $\Omega$ without deriving a suitable
representation \eqref{eq:linear_system_additive_disturbance}, i.e. a suitable
$\mStateInputBound$, first. Using the approximate model and LQR design, we choose $K_\Omega=(-4.12~~-5.32)$.
Based on $N_s = 600$ uniformly sampled measurements $\{x_i,u_i,y_i\}_{i=1}^{N_s}$ from the real (but unknown)
system, we generate the robust scenario design problem \eqref{eq:scenario_based_omega_opt}
with scenarios $\tilde w_i =  y_i - \begin{psmallmatrix}
	1 & 0.1 \\
	-0.23 & 0.78
\end{psmallmatrix} x_i - \begin{psmallmatrix}
	0 \\
	0.1
\end{psmallmatrix} u_i$. 
Solving \eqref{eq:scenario_based_omega_opt} yields that
$\Omega = \{x | x^\top P x \leq 1 \}$ with
$P = \begin{psmallmatrix} 53.95 &  11.47 \\ 11.47 & 14.55\end{psmallmatrix}$
fulfills \eqref{eq:scenario_based_omega_opt_lmi} for all possible $\tilde w_i \in \mStateInputBound$
with probability $0.97$ according to Proposition \ref{prop:scenario_based_omega}.
For the MPSC scheme, we use a horizon $N=20$
and the terminal safe set $\mSafeSet_f = \Omega$ as described in Remark
\ref{rem:practical_design_procedure}.

As learning signal we use
$u_\LL(k) = 2\sin(0.01 \pi k ) + 0.5\sin(0.12 \pi k)$ with the goal of generating
informative measurements according to \cite{ljung1998system}.

A closed-loop simulation with initial condition $x(0)=(-0.7,1)^\top$ under application
of Algorithm~\ref{alg:MPSC} is illustrated in Figure \ref{fig:closed_loop_heat}
with the corresponding safe set. As desired, the safety controller modifies the
proposed input signal $u_\LL$ only as the system state approaches a neighborhood of
the safe set boundary where the next state would leave the safe set (indicated in red color). 
The pure learning-based trajectory (dotted line, in Figure \ref{fig:closed_loop_heat}), in contrast, would have
violated state constraints already in the first time steps.

Using a similiar configuration with planning horizon $N=10$, we now
iteratively enlarge the safe set based on previously calculated nominal state trajectories
at each time step by following Corollary \ref{cor:enlargement_of_terminal_set_tube}.
Samples of the nominal and overall terminal set at different time steps are shown
in Figure \ref{fig:iterative_enlargement_of_Sf}.
After $k=115$ time steps, a significant portion of the state space is already covered by the safe terminal
set.


\begin{figure}
	\centering
	\includegraphics[width=\linewidth]{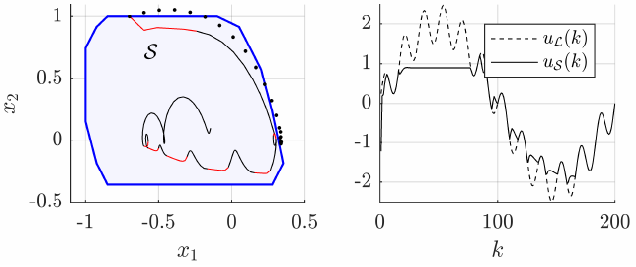}
	\caption{\textbf{Left:} Closed loop simulation under Algorithm~\ref{alg:MPSC}, starting
	from $x(0)=[-0.7,1]$.
	Red color indicates states, for which $\mNormGen{\tilde u - u_\LL}>0$. The dotted 
	line shows the first $20$ time steps of the closed-loop trajectory, resulting
	from application of $u_\LL$ without the MPSC scheme.
	\textbf{Right:}
	Learning-based control input sequence and applied controller sequence of MPSC scheme.}
	\label{fig:closed_loop_heat}
\end{figure}

\begin{figure}
	\centering
	\includegraphics[width=\linewidth]{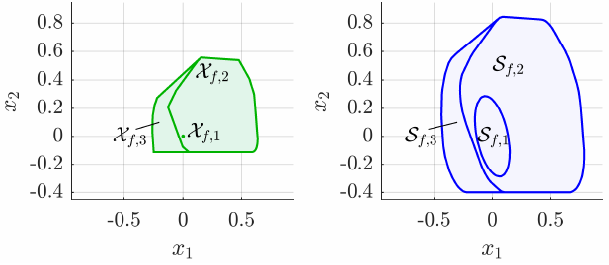}
	\caption{\textbf{Left:} Iterative enlargement of the nominal
	terminal set \eqref{eq:enlargement_of_terminal_set_tube},
	which is shown at times $k_1=0,k_2=100,k_3=115$. \textbf{Right:}
	Resulting safe \emph{terminal} set
	$\mSafeSet_f=\mathcal X_f^{\mathcal M(k)}\oplus\Omega$ corresponding
	to the nominal terminal sets at times $k_1=0,k_2=100,k_3=115$.}
	\label{fig:iterative_enlargement_of_Sf}
\end{figure}

\section{Conclusion}\label{sec:conclusion}
The paper has addressed the problem of safe learning-based control
by means of a model predictive safety certification
scheme. The proposed scheme allows for enhancing any potentially unsafe learning-based
control strategy with safety guarantees and can be combined with any
known safe set.
By relying on robust MPC methods, the presented concept is amenable for application
to large-scale systems with similar offline computational complexity as e.g.
ellipsoidal safe set approximations. Using a parameter-free
scenario-based design procedure, it was illustrated how the design steps can be performed
based on available data and how to reduce conservatism of the MPSC scheme over time by
making use of generated closed-loop data.

\bibliography{bibliography.bib}
\bibliographystyle{IEEEtran}

\end{document}